\documentclass[12pt,onecolumn,oneside,letterpaper]{IEEEtran}

\usepackage{cite}
\usepackage[final]{graphicx}
\usepackage{amsmath}
\usepackage{times}
\usepackage{latexsym}
\usepackage{graphicx}
\usepackage{bm}
\usepackage{amssymb}
\usepackage[center]{caption2}
\usepackage{stfloats}
\usepackage{cases}
\usepackage{array}
\usepackage{setspace}
\usepackage{fancyhdr}
\usepackage{fancyhdr}
\usepackage{algorithm}
\usepackage{algpseudocode}
\usepackage{textcomp}
\usepackage{wasysym}
\usepackage{balance} 
\usepackage{flushend} 
\usepackage{url}
\usepackage{xcolor}

\allowdisplaybreaks[4]

\newcaptionstyle{mystyle1}{%
  \centering TABLE  \captiontext \par}
\captionstyle{mystyle1}
\newcaptionstyle{mystyle2}{%
  \captionlabel $.$ \: \captiontext \par}
\captionstyle{mystyle2}
\newcaptionstyle{mystyle3}{%
  \captionlabel $.$ \: \captiontext \par}
\captionstyle{mystyle3}

\newtheorem{theorem}{Theorem}

\begin{document}

\title{Joint Transmitter and Receiver Design \\ for Pattern Division Multiple Access}


\author{Yanxiang~Jiang,~\IEEEmembership{Member,~IEEE}, 
Peng~Li, Zhiguo~Ding,~\IEEEmembership{Senior~Member,~IEEE}, Fuchun~Zheng,~\IEEEmembership{Senior~Member,~IEEE}, Miaoli~Ma, and~Xiaohu~You,~\IEEEmembership{Fellow,~IEEE}
\thanks{Manuscript received \today.}
\thanks{Y. Jiang, P. Li, M. Ma,  and X. You are with the National Mobile Communications Research Laboratory,
Southeast University, Nanjing 210096, China (e-mail: \{yxjiang, xhyu\}@seu.edu.cn).}
\thanks{Z. Ding is with the School of Computing and Communications, Lancaster
University, Lancaster LA1 4YW, U.K. (e-mail: z.ding@lancaster.ac.uk).}
\thanks{F. Zheng is with the National Mobile Communications Research Laboratory, Southeast University, Nanjing 210096, China, and the Department of Electronic Engineering, University of York, York, YO10 5DD, U.K. (e-mail: fzheng@seu.edu.cn).}
}

\maketitle

\begin{abstract}
In this paper, a joint transmitter and receiver design for pattern division multiple access (PDMA) is proposed. At the transmitter, pattern mapping utilizes power allocation to improve the overall sum rate, and beam allocation to enhance the access connectivity. At the receiver, hybrid detection utilizes a spatial filter to suppress the inter-beam interference caused by beam domain multiplexing, and successive interference cancellation to remove the intra-beam interference caused by power domain multiplexing. Furthermore, we propose a PDMA joint design approach to optimize pattern mapping based on both the power domain and beam domain. The optimization of power allocation is achieved by maximizing the overall sum rate, and the corresponding optimization problem is shown to be convex theoretically. The optimization of beam allocation is achieved by minimizing the maximum of the inner product of any two beam allocation vectors, and an effective dimension reduction method is proposed through the analysis of pattern structure and proper mathematical manipulations. Simulation results show that the proposed PDMA approach outperforms the orthogonal multiple access and power-domain non-orthogonal multiple access approaches even without any optimization of pattern mapping, and that the optimization of beam allocation yields a significant performance improvement than the optimization of power allocation.
\end{abstract}


\newpage

\section{Introduction}

    With the new challenges of explosive mobile data growth, tremendous increase in the number of connected devices, and continuous emergence of new service requirements, future communication systems with high spectral efficiency are needed.
    In order to efficiently support unprecedented requirements for system sum rate and access connectivity, researchers from both industry and academia are focusing on the design of next-generation multiple access techniques, particularly non-orthogonal multiple access (NOMA) \cite{Ding2017, Andrews2014What}.

    In mobile communications systems, the design of multiple access schemes is of great importance to increase the  sum rate in a cost-effective manner.
    In general, multiple access schemes can be classified into orthogonal and non-orthogonal ones based on the way wireless resources are allocated to the users.
    Orthogonal multiple access (OMA) schemes, such as orthogonal frequency division multiple access (OFDMA) in downlink and single-carrier frequency division multiple access (SC-FDMA) in uplink, are adopted in the 4G mobile communication systems such as Long-Term Evolution (LTE) and LTE-Advanced (LTE-A) \cite{Ghosh2010LTE}.
    In order to attain  further enhancements in  sum rate and access connectivity, more advanced multiple access schemes need to be developed.
    Actually, NOMA schemes are optimal in the sense of achieving the capacity region of the broadcast channel \cite{Caire2003On}.
    In NOMA schemes, multi-user signals are superposed in the same time and frequency  resources via code domain and/or power domain multiplexing at the transmitter, and separated at the receiver by multi-user detection based on successive interference cancellation (SIC) or message passing algorithm (MPA).
    Recently, several representative NOMA schemes have been proposed, such as power-domain NOMA (PD-NOMA), and sparse code multiple access (SCMA).
    PD-NOMA was introduced in \cite{Saito2013Non} by using superposition coding at the transmitter and SIC at the receiver, which lays the foundation for NOMA when a single resurce block, such as OFDM subcarrier, is available.
    SCMA was proposed in \cite{Nikopour2013Sparse} by efficiently using multiple resource blocks available in the system and mapping bit streams  directly to sparse codewords, which is thus amenable  to the  use of MPA with acceptable complexity \cite{Au2014Uplink}.

    Different from the above mentioned NOMA schemes, pattern division multiple access (PDMA) adopts PDMA patterns to separate user signals at the transmitter, where the PDMA patterns can be realized in multiple resource domains \cite{Dai2014Successive, Chen2017Pattern}.
   By means of multiple-domain multiplexing, PDMA can make the best of wireless resources to increase the sum rate with affordable computational complexity.
    Quasi-orthogonal space-time block codes in spatial domain were utilized in \cite{Mao2016Pattern} to realize resource multiplexing, where the dimension of pattern matrix scales with the number of transmitter antennas.
    The performance of cooperative PDMA was analyzed in \cite{Tang2017Performance}, which was shown to outperform cooperative OMA in terms of sum rate when the same target data rate requirement was assigned to users.
    Recently, more and more advanced receivers have been investigated for PDMA.
    Iterative detection and decoding algorithm was proposed in \cite{Ren2016Advanced}, and SIC iterative processing based on minimum mean square error (MMSE) detection and channel decoding was proposed in \cite{Kong2016Multiuser}.
    The pattern design was studied in \cite{Ren2016Pattern} for massive machine type communication (mMTC) as well as enhanced mobile broadband (eMBB) deployment scenarios, and several non-optimal design criteria were proposed.
    However, little has been  reported on some joint design for the transmitter and receiver yet, especially based on multiple-domain multiplexing.
    Although PDMA is a promising candidate for future communication systems,  comprehensive and thorough researches are still needed on the pattern and transceiver design.

    Furthermore, as one of the key technologies of future mobile communications systems, large-scale antenna arrays (LSA) have been put forward to significantly improve the system sum rate with extra degrees of freedom which facilitate transmit diversity and spatial multiplexing gains \cite{Andrews2014What, Marzetta2010Noncooperative, Wang2016An}.
    Facing a massive number of connected devices, LSA can provide sufficient spatial resources.
    More recently, the application of LSA to NOMA has been receiving growing attention for further performance improvement \cite{Zhang2015Performance, Ding2016Design}.

    Motivated by the aforementioned discussions, we propose a joint transmitter and receiver design  for PDMA.
    The proposed approach is designed based on both the power domain and beam domain in a joint manner.
    Pattern mapping at the transmitter utilizes power allocation and beam allocation to superpose user signals, while hybrid detection at the receiver  employs  a spatial filter (SF) and SIC  to separate the superposed multiple-domain signals.
    Furthermore, the optimization of  pattern mapping is investigated.
    By theoretically proving the convexity of the corresponding sum rate maximization problem, a globally optimal power allocation policy can be readily obtained.
    Through the analysis of pattern structure and proper mathematical manipulations, an effective dimension reduction method is proposed to solve the challenging optimization problem concerning beam allocation.

    The rest of the paper is organized as follows.
    The system model is described in Section II.
    The proposed PDMA joint design approach including pattern mapping at the transmitter and hybrid detection at the receiver is presented in Section III.
    The optimization of pattern mapping including power allocation and beam allocation is presented in Section IV.
    Simulation results are shown in Section V.
    Final conclusions are drawn in Section VI.

\section{System Model}

    In this paper, we consider a downlink transmission scenario with one BS communicating with multiple users.
    The BS is equipped with a LSA, where a finite number of antennas cooperate with each other and form an antenna cluster (AC) to fully exploit the cooperation gain \cite{Jiang2015Optimal}.
    Assume that there are multiple ACs located in the BS, each AC equipped with \(N_T\) antennas forms \(N\) beams with \(N_T \geq N\), and
    each user has \(N_R\) antennas.
     {For the ease of the following description and analysis,  all the users in the corresponding AC coverage are assumed to
    constitute a user group (UG).}
	Assume each UG contains \(K\) users and an AC covers a UG with \(N_T \le KN_R\) and \(N \le K \le 2^N-1\) \cite{Dai2014Successive}.
    In this case, one beam will have to support more than one user,
     {i.e., some users in the same UG will share one beam.}
    Without loss of generality, we simplify the scenario into the case where an AC communicates with a UG.

    Let \(\boldsymbol{G}_k \in {\mathbb{C}^{N_R \times N_T}}\) denote the channel matrix between the considered AC and the \(k\)-th user in the considered UG.
    Assume that the BS has perfect channel state information (CSI).
    BF at the transmitter is generated based on the CSIs of $N$ target users from the UG.
    Let $n_{\rm{tar}}$ denote one of the target users covered by the \(n\)-th beam, and the set of target users is expressed as \(\Omega  = \left\{ {1_{\rm{tar}}}\text{, }{2_{\rm{tar}}}\text{, } \cdots \text{, }{n_{\rm{tar}}}\text{, } \cdots \text{, }{N_{\rm{tar}}} \right\}\).
    Let ${\boldsymbol{f}_n \in {\mathbb{C}^{N_T \times 1}}}$ denote the BF vector of the \(n\)-th beam, which is generated based on the CSI of the user $n_{\rm{tar}}$.
    Let $\boldsymbol{F} \in {\mathbb{C}^{N_T \times N}}$ denote the BF matrix and it can be expressed as \(\boldsymbol{F}=\left[ {\boldsymbol{f}_1\text{, }\boldsymbol{f}_2\text{, } \cdots \text{, } \boldsymbol{f}_n\text{, } \cdots \text{, } \boldsymbol{f}_N} \right]\).
     {According to different user requirements, the BF matrix \(\boldsymbol{F}\) can be constructed based on the minimum mean square error (MMSE) or zero-forcing (ZF) criteria.}

    At the transmitter, let \(\boldsymbol{t} \in {\mathbb{C}^{N \times 1}}\) denote the superposed signal vector after pattern mapping, and the details for the design of $\boldsymbol{t}$ will be provided in the next section.
    Let \(\boldsymbol{x} \in {\mathbb{C}^{N_T \times 1}}\) denote the transmit signal vector from the AC, and it can be expressed as follows
\begin{equation}\begin{split}\label{transmitted_signal}
\boldsymbol{x} = \boldsymbol{F} \boldsymbol{t}\text{.}
\end{split}\end{equation}
At the receiver, let \(\boldsymbol{y}_k \in {\mathbb{C}^{N_R \times 1}}\) denote the received signal vector for the \(k\)-th user.
    Then, it can be expressed as follows
\begin{equation}\begin{split}\label{received_signal}
{\boldsymbol{y}_k} = {\boldsymbol{G}_k}\boldsymbol{x} + {\boldsymbol{w}_k}{,} \ k = 1, 2, \cdots, K,
\end{split}\end{equation}
    where \(\boldsymbol{w}_k \sim \mathcal{CN}(\bold{0}\text{, } \sigma_k^2{\boldsymbol{I}_{N_R})}\) denotes the additive white Gaussian noise vector whose elements have zero mean and variance $\sigma_k^2$.

    In this paper, by constructing the superposed signal vector \(\boldsymbol{t}\)
at the transmitter and detecting the received signal vector \(\boldsymbol{y}_k\) at
the receiver, a PDMA joint design approach based on both the power domain and beam domain is proposed.

\section{The Proposed PDMA Joint Design Approach}

In this section, we present our proposed PDMA joint design approach  as illustrated in Fig. \ref{scheme}, where \(K\) user signals are superposed upon \(N\) beams after the process of pattern mapping at the transmitter and detected at the receiver by means of hybrid detection.

\begin{figure}[!t]
\centering
\includegraphics[width=0.9\textwidth]{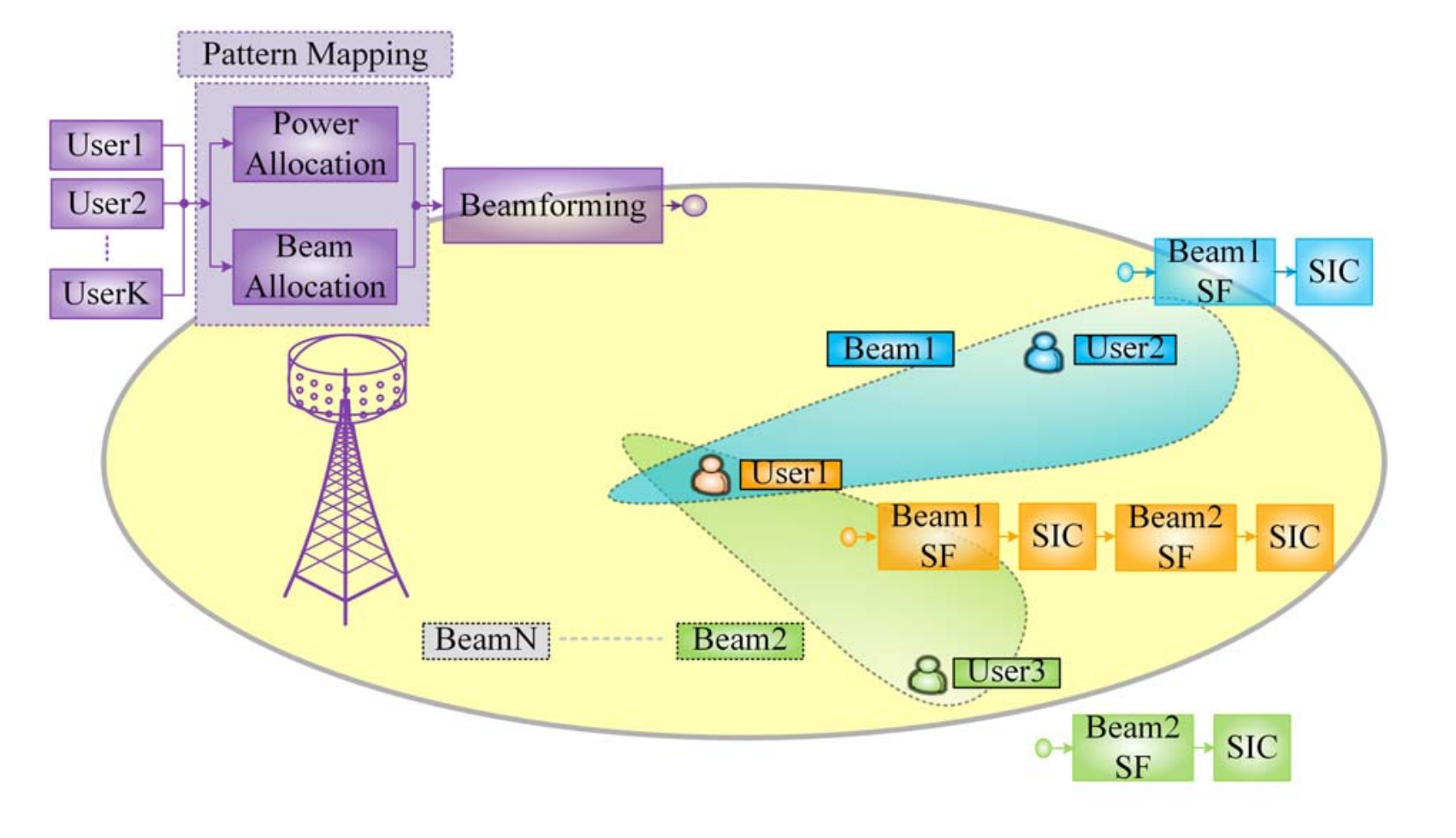}
\caption{Illustration of the proposed PDMA joint design approach.}
\label{scheme}
\end{figure}

\subsection{Pattern Mapping at the Transmitter}

The multiple-domain multiplexing is one of  the key aspects of pattern mapping in PDMA.
In general, the power domain is chosen as the fundamental  multiplexing domain in pattern mapping \cite{Dai2014Successive, Chen2017Pattern}.
However, when only power domain is utilized, the stringent requirement in power allocation policy restricts its generalization 
to various scenarios.
    In fact, the combination of multiple domains can make the most of wireless resources and generalize PDMA to various application scenarios.
    By considering LSA, multiple beams in downlink can serve as spatial resources.
    Specifically, beams are shared by multiple users with different power, and the corresponding allocation policy depends on pattern mapping.

    In this paper, pattern mapping of our proposed PDMA approach is designed based on the combination of power domain and beam domain.
    Power resources are chosen as the fundamental multiplexing domain and beam resources as the key multiplexing domain in the pattern mapping.
As shown in Fig. \ref{scheme}, pattern mapping at the transmitter utilizes power allocation and beam allocation to superpose multiple-domain signals.
    Let \(\boldsymbol{s} \in {\mathbb{C}^{K \times 1}}\) denote the transmit symbol vector for the UG with \(\boldsymbol{s} \sim \mathcal{CN}(\bold{0}\text{, }\boldsymbol{I}_K)\) and \({\left[ \boldsymbol{s} \right]_{k}} = {s_{k}}\) the transmit symbol for the \(k\)-th user.
    Let \(\boldsymbol{P} \in {\mathbb{R}^{N \times K}}\) denote the power allocation matrix and \({\left[ \boldsymbol{P} \right]_{nk}} = { {p_{nk}}}\) the transmit power allocated to the \(k\)-th user in the \(n\)-th beam.
    Assume that each beam is allocated the same power $P_{\text{b}}$, i.e., $\sum\limits_{k = 1}^K {{p_{nk}}}=P_{\text{b}}, \ n = 1, 2, \cdots, N$.
    Let \(\boldsymbol{B} \in {{[0,1]}^{N \times K}}\) denote the beam allocation matrix and \({\left[ \boldsymbol{B} \right]_{nk}} = {b_{nk}}\) the beam allocation indicator for the \(k\)-th user in the \(n\)-th beam with $b_{nk}=1$ if the $k$-th user is covered by the $n$-th beam and $b_{nk}=0$ otherwise.
    Correspondingly, the superposed signal vector after pattern mapping can be expressed as follows
\begin{equation}\begin{split}\label{superposedSignal}
\boldsymbol{t} = \left( {\boldsymbol{B} \circ {\boldsymbol{P}}^{1/2}} \right) \boldsymbol{s}\text{,}
\end{split}\end{equation}
    where $\circ$ denotes the operation of  Hadamard product, and \({\left[ \boldsymbol{t} \right]_{n}} = t_n = \sum\limits_{k = 1}^K {{b_{nk}}\sqrt {{p_{nk}}} {s_k}}, n = 1, 2, \cdots, N \).

\subsection{Hybrid Detection at the Receiver}

    For the received signals, we propose to use hybrid detection as illustrated in Fig. \ref{scheme}, where SF is used to suppress the inter-beam interference caused by beam domain multiplexing and SIC is then used  to remove the intra-beam interference caused by power domain multiplexing.

    {Firstly, SF is performed for the received signal vector to suppress the inter-beam interference.}
    Let \({\boldsymbol{V}_{k}} \in {\mathbb{C}^{N_R \times N}}\) denote the SF matrix for the \(k\)-th user and \({\boldsymbol{v}_{nk}}\) the \(n\)-th column vector of \({\boldsymbol{V}_{k}}\).
    After the procedure of SF, the scalar received signal \(z_{nk}\) can be expressed as follows
\begin{equation}\label{scalar_received_signal}
{z_{nk}} = \boldsymbol{v}_{nk}^H{\boldsymbol{y}_k}{\rm{ = }}\boldsymbol{v}_{nk}^H{\boldsymbol{G}_k}{\boldsymbol{f}_n}{t_n} + \boldsymbol{v}_{nk}^H{\boldsymbol{G}_k}\sum\limits_{n'=1\hfill\atop n'\ne n\hfill}^N {{\boldsymbol{f}_{n'}}{t_{n'}}}  + \boldsymbol{v}_{nk}^H{\boldsymbol{w}_k}, \ n=1,2,\cdots,N, \ k = 1,2,\cdots, K,
\end{equation}
    where the first term of the right hand denotes the combination of the desired information and intra-beam interference, and  the other two terms denote the inter-beam interference and noise, respectively.
     {According to different user requirements, the SF matrix \({{\boldsymbol{V}_{k}}}\) can be constructed based on the MMSE or ZF criteria.}
    If \({\boldsymbol{V}_{k}}\) is constructed based on the MMSE criterion,  it can then be expressed as follows  \cite{Higuchi2013Non}
\begin{equation}\begin{split}\label{SF_MMSE}
{\boldsymbol{V}_{k}} & = \mathop {\min }\limits_{{\boldsymbol{\tilde V}_k}} {\mathbb E}\left\{ {\left\| {\boldsymbol{t} - \boldsymbol{\tilde V}_k^H{\boldsymbol{y}_k}} \right\|_2^2} \right\} \\
& = {\left( {{\boldsymbol{G}_k}\boldsymbol{FA}{\boldsymbol{F}^H}\boldsymbol{G}_k^H + \sigma _k^2{{\bf{I}}_{N_R}}} \right)^{ - 1}}{\boldsymbol{G}_k}\boldsymbol{FA}{,}
\ k=1,2,\cdots,K,
\end{split}\end{equation}
    where \(\boldsymbol{A} \buildrel \Delta \over = \mathbb{E}\left\{ {\boldsymbol{t}{\boldsymbol{t}^H}} \right\}\) with its element being expressed as \({\left[ \boldsymbol{A} \right]_{ij}} = \sum\limits_{k = 1}^K {{b_{ik}}{b_{jk}}\sqrt {{p_{ik}}{p_{jk}}} } \).
    Note that the computational complexity of the matrix inversion involved in (\ref{SF_MMSE}) is approximately $O(N_R^3)$.
    If \({\boldsymbol{V}_{k}}\) is constructed based on the ZF criterion, it can then be expressed as follows
\begin{equation}\begin{split}\label{SF_ZF}
{\boldsymbol{V}_k} = {\boldsymbol{G}_k}\boldsymbol{F}{\left( {{\boldsymbol{F}^H}\boldsymbol{G}_k^H{\boldsymbol{G}_k}\boldsymbol{F}} \right)^{ - 1}}, \ k=1,2,\cdots,K{.}
\end{split}\end{equation}
    The computational complexity of the matrix inversion involved in (\ref{SF_ZF}) is approximately $O(N^3)$.
    In a practical communication  scenario, both $N_R$ and $N$ are not so large.
    Therefore, the computational complexity of the SF procedure can be affordable.

    {Secondly, normalization is applied to the scalar received signal to meet the implementation condition of SIC.}
    For a fixed channel realization, the inter-beam interference and noise term in \(z_{nk}\) is   {assumed independently Gaussian distributed} with  mean zero and variance $\sum\limits_{n' = 1\hfill\atop n' \ne n\hfill}^N {{{\left| {\boldsymbol{v}_{nk}^H{\boldsymbol{G}_k}{\boldsymbol{f}_{n'}}} \right|}^2}P_{\text{b}}}  + \sigma _k^2{{\left\| {{\boldsymbol{v}_{nk}}} \right\|}^2}$.
    Correspondingly, the expression in (\ref{scalar_received_signal}) can be reshaped through normalization as follows
\begin{equation}\begin{split}\label{normalized}
{z'_{nk}} = {h_{nk}}\sum\limits_{k' = 1}^K {{b_{nk'}}\sqrt {{p_{nk'}}} {s_{k'}}}  + {q_{nk}}{,} \ n=1,2,\cdots,N, \ k =1,2,\cdots,K,
\end{split}\end{equation}
    where \(q_{nk}\) denotes the sum of the inter-beam interference and noise after normalization with \(\mathbb{E}\left[ {{{\left| {{q_{nk}}} \right|}^2}} \right] = 1\), \(h_{nk}\)  the equivalent normalized channel gain between the \(k\)-th user and the AC which can be expressed as follows
\begin{equation}\label{h_nk}
{h_{nk}} =  {\frac{1}{\sqrt {\sum\limits_{n' \ne n, n' = 1 
}^N {{{\left| {\boldsymbol{v}_{nk}^H{\boldsymbol{G}_k}{\boldsymbol{f}_{n'}}} \right|}^2}P_{\text{b}}}  + \sigma _k^2{{\left\| {{\boldsymbol{v}_{nk}}} \right\|}^2}}}}{\boldsymbol{v}_{nk}^H{\boldsymbol{G}_k}{\boldsymbol{f}_n}}, \
 n=1,2,\cdots,N, \ k =1,2,\cdots,K{.}
\end{equation}
    Let $\boldsymbol{H} \in {\mathbb{R}^{N \times K}}$ denote the equivalent normalized channel gain matrix between the considered UG and AC, and it satisfies ${\left[ \boldsymbol{H} \right]_{nk}} = {h_{nk}}$.
    Then, the multiple-input multiple-out (MIMO) channel between the \(k\)-th user and the considered AC can be  degraded into  single-input single-out (SISO) channels after the normalization \cite{Higuchi2013Non}, which facilitates the implementation of  SIC.

    {Finally, SIC is applied to the normalized scalar received signal to remove the intra-beam interference,}
    and it is  used to decode symbols iteratively by subtracting the detected symbols of weak users first to facilitate the detection of strong users.
    Without loss of generality, we assume that the \(K\) users are sorted in an ascending order of normalized channel gain \(h_{nk}\) with respect to the index number \(k\) for any index number \(n\).
    For instance, \(\left| {{h_{nk}}} \right| \le \left| {{h_{nk'}}} \right|\) holds for any \(n=1,2, \cdots, N\) if \(1 \le k \le k' \le K\).
    Consequently, the \(k'\)-th user can correctly decode the signal symbol in spite of the interference of the \(k\)-th user in the \(n\)-th beam.
    Correspondingly, the signal to interference plus noise ratio (SINR) of the \(k\)-th user in the \(n\)-th beam can be expressed as follows
\begin{equation}\begin{split}\label{SINR}
{\gamma _{nk}} = {\rm{ }}\left\{ {\begin{array}{*{20}{l}}
{\frac{{{\left| h_{nk}\right|}^2{b_{nk}}{p_{nk}}}}{{1 + {\left| h_{nk}\right|}^2\sum\limits_{k' = k + 1}^K {{b_{nk'}}{p_{nk'}}} }}{, } }& {n=1,2,\cdots, N, \ k = 1,2{, } \cdots {, }K - 1}{, }\\
{{\left| h_{nk}\right|}^2{b_{nk}}{p_{nk}}{, }} & {n=1,2,\cdots, N, \ k = K}{.}
\end{array}} \right.
\end{split}\end{equation}
    And  the overall sum rate of the AC can then be expressed as follows
\begin{equation}\begin{split}\label{sumrate}
{R_{{\rm{sum}}}} = \sum\limits_{n = 1}^N {\sum\limits_{k = 1}^K {{{\log }_2}\left( {1 + {\gamma _{nk}}} \right)} }\text{.}
\end{split}\end{equation}

    Note that the performance gain of our proposed approach profits from  enhanced access connectivity supported by the multiple-domain multiplexing and strong SINR ensured by the hybrid detection.
    The hybrid detection at the receiver is designed based on the multiple-domain multiplexing at the transmitter, and  the corresponding optimization of pattern mapping at the transmitter can further improve the system performance.

\section{Optimization of Pattern Mapping}

	Pattern mapping is optimized based on both the power domain and beam domain in the proposed PDMA joint design approach.
    Power resources act as the fundamental multiplexing domain of the pattern, and beam resources act as  the key multiplexing domain of the pattern.
     {Given the integer beam allocation indicator $b_{nk}$, the optimization problem of the overall sum rate maximization falls into the scope of combinatorial programming, which is hard to be solved directly. Just as we pointed out in \cite{Jiang_TVT16}, a brute force
approach is generally required for obtaining a global optimal
solution. However, such an approach has an exponential complexity with respect to (w.r.t.) the number of beams and the number of users, and it is computationally impracticable even for a small-size system.
In our previous work in \cite{Jiang_VTC17S}, by merging the beam allocation indicator $b_{nk}$ and the transmit power $p_{nk}$ into one variable, the original 2-D optimization problem of pattern mapping can be simplified  to
a 1-D optimization problem. However, by using that approach, outlier may exist in the obtained pattern mapping results.
Therefore, in this paper, we propose an efficient  approach to solve this challenging problem.} 
    Pattern design of the proposed PDMA joint design approach involves the optimization of  power allocation and that of beam allocation.

\subsection{Optimization of Power Allocation}

    In this subsection, the optimal power allocation  of the proposed PDMA approach is investigated.
    In general, NOMA schemes are used to increase system rates in  future mobile communications systems.
    The optimization of power allocation  can be formulated as a sum rate maximization problem.
    Besides, BF at the transmitter is generated based on the CSI of the target user, and hence the target user should be allocated with nonzero power to ensure its availability.
    Let ${\delta _{nk}}$ denote the lower bound of the transmit power allocated to the $k$-th user in the $n$-th beam, and it can be expressed as follows
\begin{equation}\begin{split}\label{delta_nk}
{\delta _{nk}} = \left\{ {\begin{array}{*{20}{l}} {\varepsilon\text{, }} &{ \text{if} \   n = 1,2,\cdots,N, \  k \in \Omega }, \\ {0\text{, } } & {\text{else} }
\text{, } \end{array}} \right.
\end{split}\end{equation}
where $\varepsilon$ denotes the slack variable for the target user.
    Therefore, the power allocation problem can be expressed as follows
\begin{equation}\begin{split}\label{OP}
\mathop {\max }\limits_{p_{nk} \atop n=1,2,\cdots,N,  k=1,2,\cdots,K} \mathop {}\nolimits^{} & \mathop {R_{\text{sum}}} = \sum\limits_{n = 1}^N {\sum\limits_{k = 1}^K {{{\log }_2}\left( {1 + {\gamma _{nk }}} \right)} } \\
\text{s.t. } \ & \text{C1 : } {p_{nk}} \ge {\delta _{nk}}\text{, } n=1,2,\cdots,N, \ k=1,2,\cdots,K, \\
\quad & \text{C2 : } \sum\limits_{n = 1}^N {\sum\limits_{k  = 1}^K {{p_{nk}}} }  \le {P_{{\text{sum}}}}{, } \\
\quad & \text{C3 : } {\log _2}\left( {1+{\gamma_{nk}}} \right) \ge {R_{nk}^{{\text{min}}}}\text{, } n=1,2,\cdots,N, \ k=1,2,\cdots,K,
\end{split}\end{equation}
    where ${P_{{\text{sum}}}}$ denotes the maximum sum transmit power for the AC, and ${R_{nk}^{{\text{min}}}}$  the minimum rate requirement for the \(k\)-th user in the \(n\)-th beam.
    Correspondingly, we have the following theorem.
    \begin{theorem}
    The optimization problem in (\ref{OP}) is convex.
    \end{theorem}
    \begin{proof}
    Please refer to the Appendix A.
    \end{proof}

     {Generally, the Karush-Kuhn-Tucker (KKT) conditions can be used to solve the  convex optimization problem in \eqref{OP}. However, the inequality constraints in the KKT approach usually makes the corresponding optimization problem hard to be solved directly. By adding a barrier function into the original objective function, the inequality constraints can be readily removed. Therefore, we adopt the barrier method \cite{Boyd2004Convex} to tackle the problem and get the optimal power allocation policy in an undemanding manner.}

\subsection{Optimization of Beam Allocation}

    The beam resource is another key multiplexing domain in our approach.
    Pattern design of the PDMA joint design approach  is mainly to optimize  beam allocation.
    In the following, the in-depth analysis on pattern structure  is presented, and then the relevant optimization of  beam allocation  is investigated.

\subsubsection{Pattern Structure}

    Beam allocation matrix $\boldsymbol{B}$ is specified as the PDMA pattern matrix. Let ${\left[ \boldsymbol{B} \right]_k} = {\boldsymbol{b}_k}$ denote the beam allocation vector of the \(k\)-th user.
    In other words, the transmit symbols for the UG are mapped into beam domain by means of  pattern, and the specific allocation of beam resources to the considered user is performed  through  the corresponding  beam allocation vector.

    The PDMA pattern matrix is utilized to assign different transmit diversity orders to users as well.
    To further analyze the pattern structure, we consider an example \(\boldsymbol{B}\) with \(N=3\) and \(K=5\) that is designed as follows
\begin{equation}\begin{split}\label{B35}
{\boldsymbol{B}} = \left[ {\begin{array}{*{20}{c}}
1&1&0&1&0\\
1&1&1&0&0\\
1&0&1&0&1
\end{array}} \right]\text{.}
\end{split}\end{equation}
    We show its factor graph in Fig. \ref{patternmatrix}.
    Each row of \(\boldsymbol{B}\) is referred to as a beam resource and each column of \(\boldsymbol{B}\) is referred to as a user.
    Therefore, the sum  of the corresponding column of \(\boldsymbol{B}\) for the considered user represents its transmit diversity order.
    The sum of the corresponding row in \(\boldsymbol{B}\) for the considered beam represents its overlap order.
    The ratio between the number of columns and the number of rows represents the overload ratio of the PDMA pattern.

\begin{figure}[!tbp]
\centering
\includegraphics[width=0.7\textwidth]{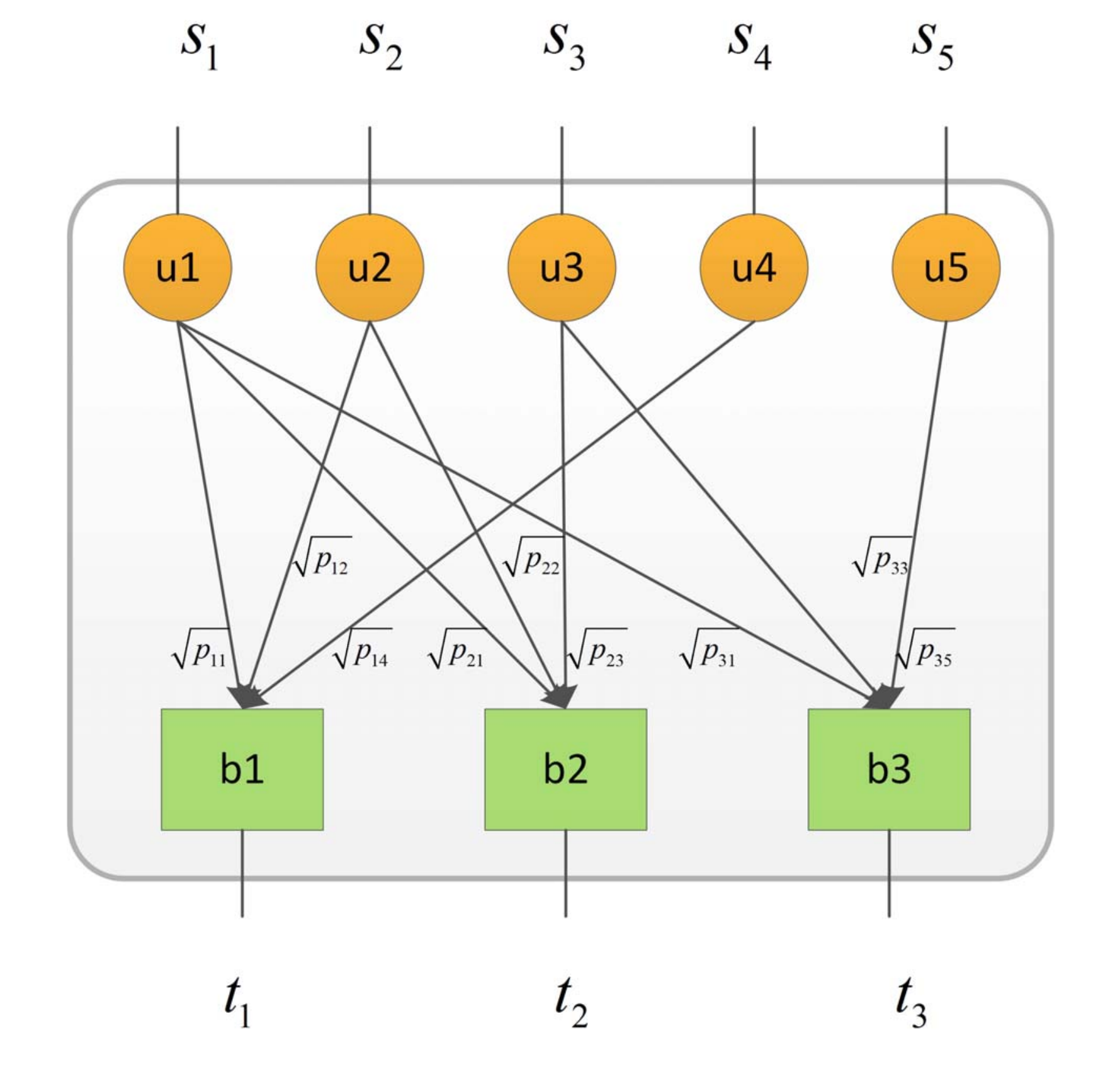}
\caption{The factor graph corresponding to the PDMA pattern with \(N=3\) and \(K=5\).}
\label{patternmatrix}
\end{figure}

    According to the above description  of the PDMA pattern, we further explore some insights on the  pattern structure.
\begin{itemize}
 \item   Firstly, the proposed PDMA approach can be treated as a more general framework of PD-NOMA  and OMA.
    Specifically, by setting all the transmit diversity orders to be one and all the overlap orders to be greater than one, the proposed PDMA approach can be transformed to a PD-NOMA one.
    On the other hand, it can be transformed to an OMA approach by setting all the transmit diversity orders and overlap orders to be one.

  \item  Secondly, the PDMA pattern can be utilized to improve  access connectivity.
    Let $\lambda$ denote the overload ratio supported by the PDMA pattern with  $\lambda=K/N$.
    By exploiting the PDMA pattern, the access connectivity in the PDMA approach can then be increased to a maximum of \(\left( {{2^N} - 1} \right)/N\) folds in comparison with the OMA approach.

  \item  Finally, the PDMA pattern enables dimension reduction for the LSA.
    Note  that the maximum transmit diversity order in pattern should be $N$ \cite{Dai2014Successive}, i.e., there is some certain user which transmits signal in all the beams.
    Therefore, for the case of narrow beamwidth, the maximum number of beams in one AC covering the corresponding UG cannot be too large.
    Specifically, a small matrix $\boldsymbol{B}$ is the common case for the proposed PDMA approach.
    Beams are generally chosen as spatial resources in the proposed approach other than antennas or space-time block codes in spatial domain \cite{Mao2016Pattern}.
    Therefore, when combined with the LSA, the dimension  of $\boldsymbol{B}$ is reduced from ${N_T} \times K$ to $N \times K$ with \(N_T \geq N\), which can further enable the computational reduction of the receiver.
\end{itemize}

\subsubsection{Optimization of Beam Allocation}

    The key  of the proposed  PDMA approach is pattern design, which can impose significant effect on system performance and detection complexity.
    On the one hand, for a given $N$, different overload ratio $\lambda$ can support varying access connectivity.
    The larger the overload ratio, the better the system performance and  the higher the detection complexity.
    On the other hand, for some certain user, larger  transmit diversity order can enable more reliable data transmission at the cost of  higher detection complexity.
    Therefore, proper pattern design is required to achieve  a good trade-off between 
    system performance and detection complexity.

    The inner product of  the beam allocation vectors of any two users indicates the number of beam resources shared by the two users.
    When the inner product is non-zero, the two users share some of the same beam resources.
    Then, the two users cannot be distinguished on  beam domain but power domain at the receiver.
     {When the inner product is zero, the two users employ  different beam resources.
    Then, they will have a large possibility not to interfere with each other on  beam domain,} which can greatly reduce the detection complexity at the receiver.
     {Therefore, by minimizing the maximum of the inner product of any two  beam allocation vectors, the inter-beam interference of any two beams can be reduced to the greatest extent and the overall sum rate can then be maximized as much as possible.
    According to the above discussions, we formulate the optimization problem of beam allocation in a heuristic way as follows}
\begin{equation}\begin{split}\label{OP_B}
\quad & \mathop {\min\max }\limits_{{\boldsymbol{b}}_k, {\boldsymbol{b}}_{k'}
} {\boldsymbol{b}}_k^T{{\boldsymbol{b}}_{k'}}\\
\quad  \text{s.t.}\ &
{   k,k' = 1,2,\cdots, K, \ k \neq k',} \\
&{\mu _1} \le {\mu _2} \le  \cdots  \le {\mu _k} \le  \cdots  \le {\mu _{K}}\text{,}
\end{split}\end{equation}
where ${\mu _k}$ denotes the transmit diversity order of the $k$-th user, and we assume that the beam allocation vectors of the $K$ users are sorted in an ascending order of their transmit diversity orders.
The optimization problem in (\ref{OP_B}) falls into the scope of combinatorial programming.
 {Just as we point out previously, direct solving  of the the optimization problem needs large computational load.}
Therefore, we commit to an effective dimension reduction method to solve this challenging problem.

By using elementary rank transformations, the matrix $\boldsymbol{B}$ can be reshaped in a partitioned structure: ${\boldsymbol{B}}' \triangleq \left[ {\bar{\boldsymbol{B}},\tilde{\boldsymbol{B}}} \right]$, where $\bar{\boldsymbol{B}} \in {\left[ {0,1} \right]^{N \times N}}$ refers to the pattern of all the users in the target user set $\Omega$, and $\tilde{\boldsymbol{B}} \in {\left[ {0,1} \right]^{N \times \left( {K - N} \right)}}$ refers to the pattern of all the other users among the UG.
Then, we can get the following theorem.
    \begin{theorem}
    If the BF matrix is constructed based on the ZF 
    or MMSE criteria and the SF matrix is also constructed based on the ZF 
    or MMSE criteria, $\bar{\boldsymbol{B}}$ can be set to  $\boldsymbol{I}$.
    \end{theorem}
    \begin{proof}
    Please refer to the Appendix B.
    \end{proof}
 {Note here that  the corresponding users in the target user set $\Omega$ will have a large possibility not to interfere with each other by setting $\bar{\boldsymbol{B}}$ to  $\boldsymbol{I}$, which apparently helps to simplify the optimization problem in \eqref{OP_B}.}
Let ${[ \tilde{\boldsymbol{B}} ]_k} = {\tilde{\boldsymbol{b}}_{k}}$ denote the vector of the \(k\)-th column of $\tilde{\boldsymbol{B}}$.
Then, according to Theorem 2,  the optimization problem in (\ref{OP_B}) can be transformed into the following equivalent form 
\begin{equation}\begin{split}\label{OP_B_2}
\quad & \mathop {\min\max }\limits_{\tilde{\boldsymbol{b}}_{k}, \tilde{\boldsymbol{b}}_{k'}} \tilde{\boldsymbol{b}}_{k}^T{\tilde{\boldsymbol{b}}_{k'}}\\
\quad  \text{s.t.} \ &
{   k,k' = 1,2,\cdots, K-N, \ k \neq k',} \\
& {\mu' _1} \le {\mu' _2} \le  \cdots  \le {\mu' _k} \le  \cdots  \le {\mu' _{K-N}}
\text{,}
\end{split}\end{equation}
where ${\mu ' _k}$ denotes the transmit diversity order of the $k$-th user concerning $\tilde{\boldsymbol{B}} $.
It can be seen that the dimension of the variable space for the optimization problem is greatly reduced from $2^{N{K}}$ to $2^{N{(K-N)}}$.
For the optimization problem in (\ref{OP_B_2}),
a small beam allocation matrix is the common case for the proposed PDMA approach, and consequently the complete enumeration method can be employed to achieve the optimal beam allocation.

 {Note that in practice there may exist inter-beam and inter-AC/UG interferences as well due to the possible beam overlap (e.g. between neighbouring beams) even when the inner product of  the beam allocation vectors of any two users is zero. However, by using our proposed approach, both the interferences and  the computational complexity can be really reduced.}
	
\section{Simulation Results}

    In this section, we evaluate the performance of our proposed PDMA joint design approach.
In the simulations, the number of antennas of each AC $N_T$ and that of each user $N_R$ are set to   \(16\) and \(4\), respectively.
    The BS is located in the cell center with {radius 800 meters.}
    It is assumed that all the users in the considered UG are distributed uniformly in their corresponding AC coverage. 
    For the propagation channel, it is assumed that the complex propagation coefficient between each antenna of the BS and that of each user is modeled as a complex small-scale fading factor timed by a large-scale fading factor, which represents geometric attenuation and shadow fading \cite{Marzetta2010Noncooperative}.
    For the small-scale fading factor, it is always assumed to be i.i.d. random variable with distribution \(\mathcal{CN}(0\text{, }1)\).
    For the large-scale fading factor, the path loss factor, the path loss exponent and the variance of log-normal shadow fading are set to  1, 3.7 and 10dB, respectively.
In the following, unless otherwise stated, we assume that {each AC  contains \(N=3\) beams,} and that the number of users varies from 4 to 7, which means that the proposed PDMA approach can perform under  different overload ratios (\(\lambda = 133\% \sim 233\%\)).

    In the simulations, the performance of the PDMA approach is firstly evaluated without any optimization on pattern mapping, i.e., simple power allocation and  beam allocation policies are adopted.
    For the simple power allocation policy \cite{Saito2013Non}, users are sorted in an ascending order with respect to the normalized channel gain.
    Let \(p_0\) denote the basic transmit power, which is allocated to the first ordered user.
    Then, the transmit power allocated to the \(k\)-th ordered user is set to  \({\mu ^{k-1}}{p_0}\), where \(\mu\) denotes the power gain factor.
    For the simple beam allocation policy, the basic pattern design criterion \cite{Dai2014Successive} is that a larger transmit diversity order is allocated to the user with smaller normalized channel gain, and vice versa.

\begin{figure}[!tbp]
\centering
\includegraphics[width=0.8\textwidth]{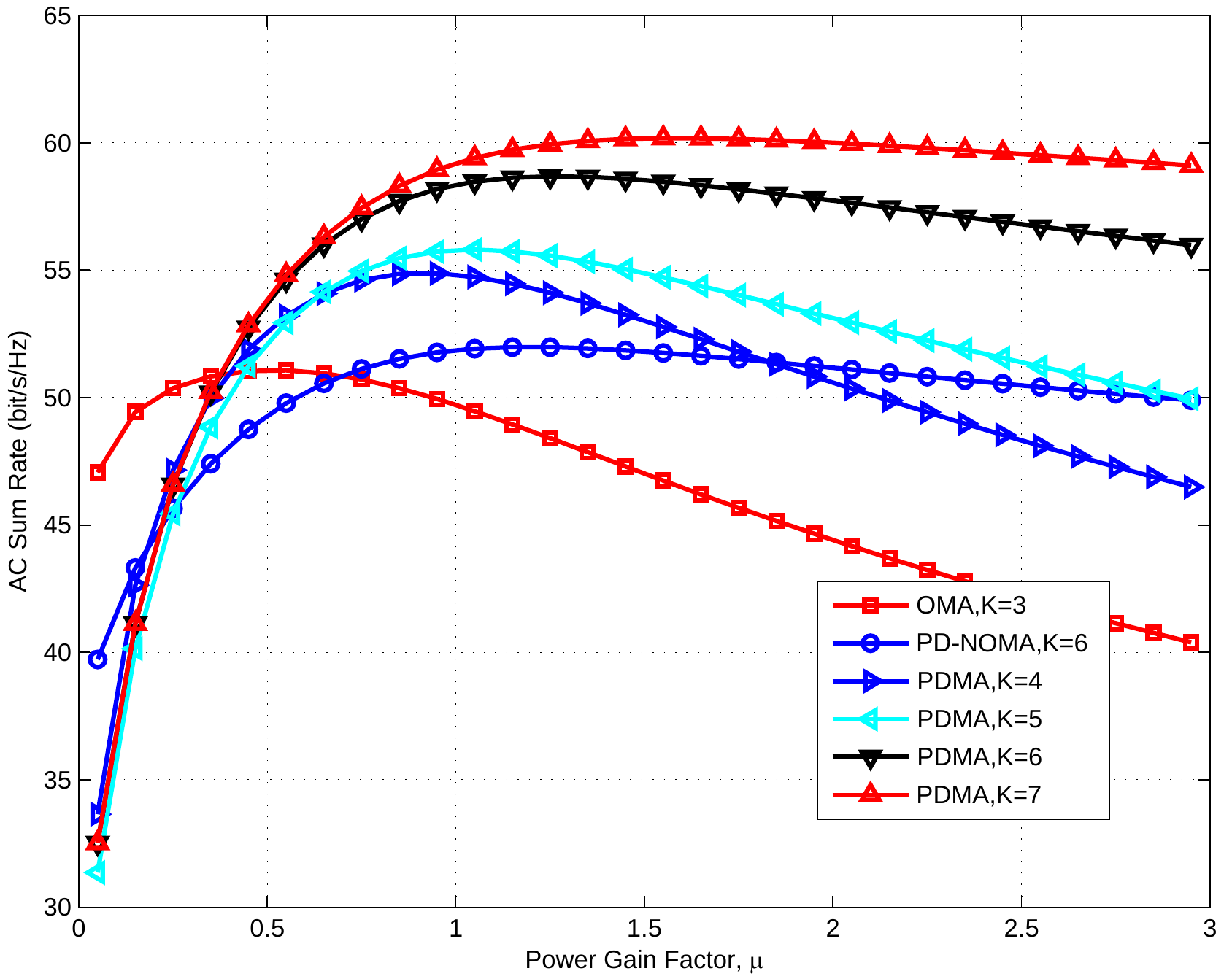}
\caption{AC sum rate versus power gain factor, where  simple power allocation and beam allocation policies are employed for the proposed PDMA approach.}
\label{simulation1}
\end{figure}

\begin{figure}[!tbp]
\centering
\includegraphics[width=0.8\textwidth]{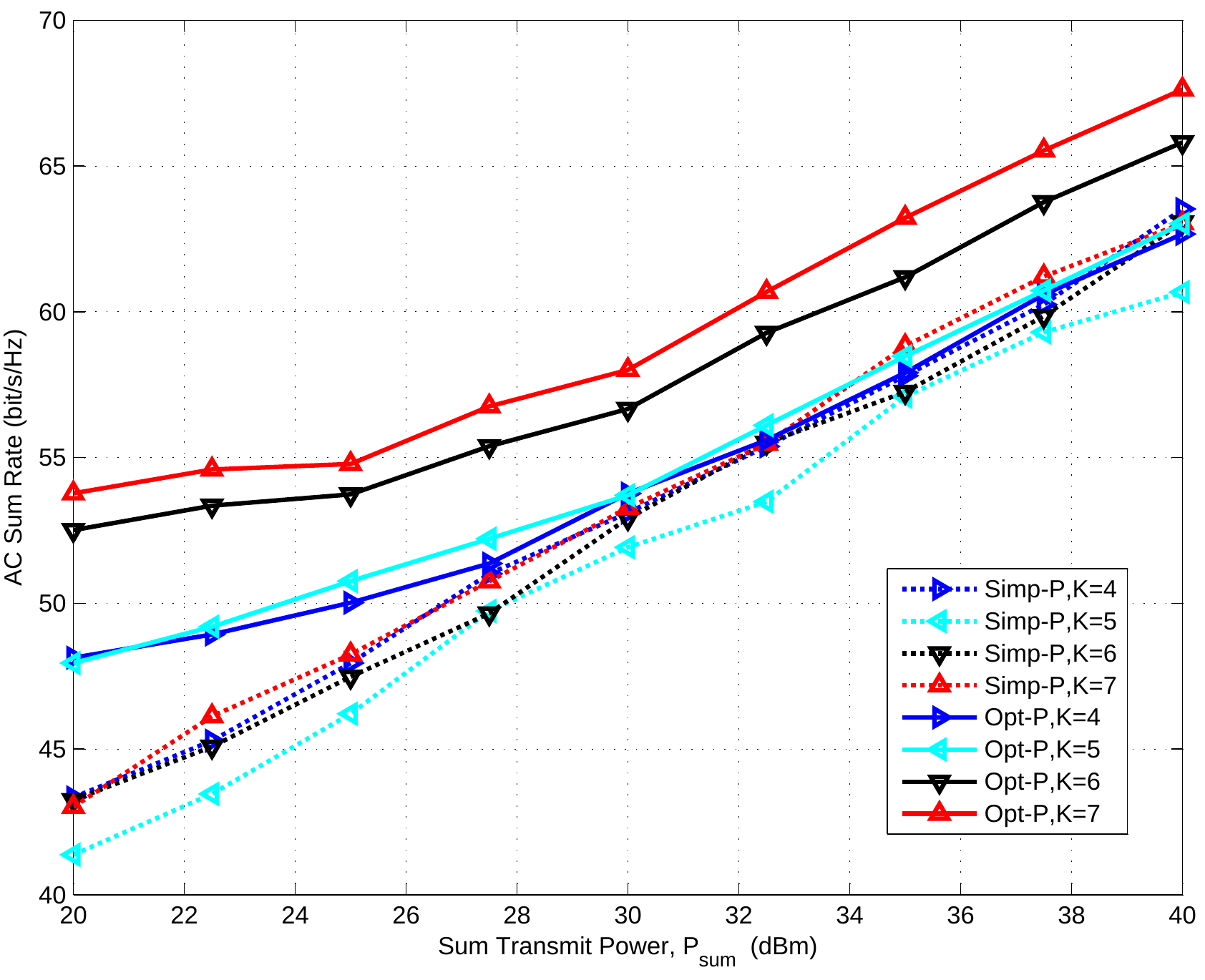}
\caption{AC sum rate versus sum transmit power, where a simple beam allocation policy is employed.}
\label{simulation2}
\end{figure}

\begin{figure}[!tbp]
\centering
\includegraphics[width=0.8\textwidth]{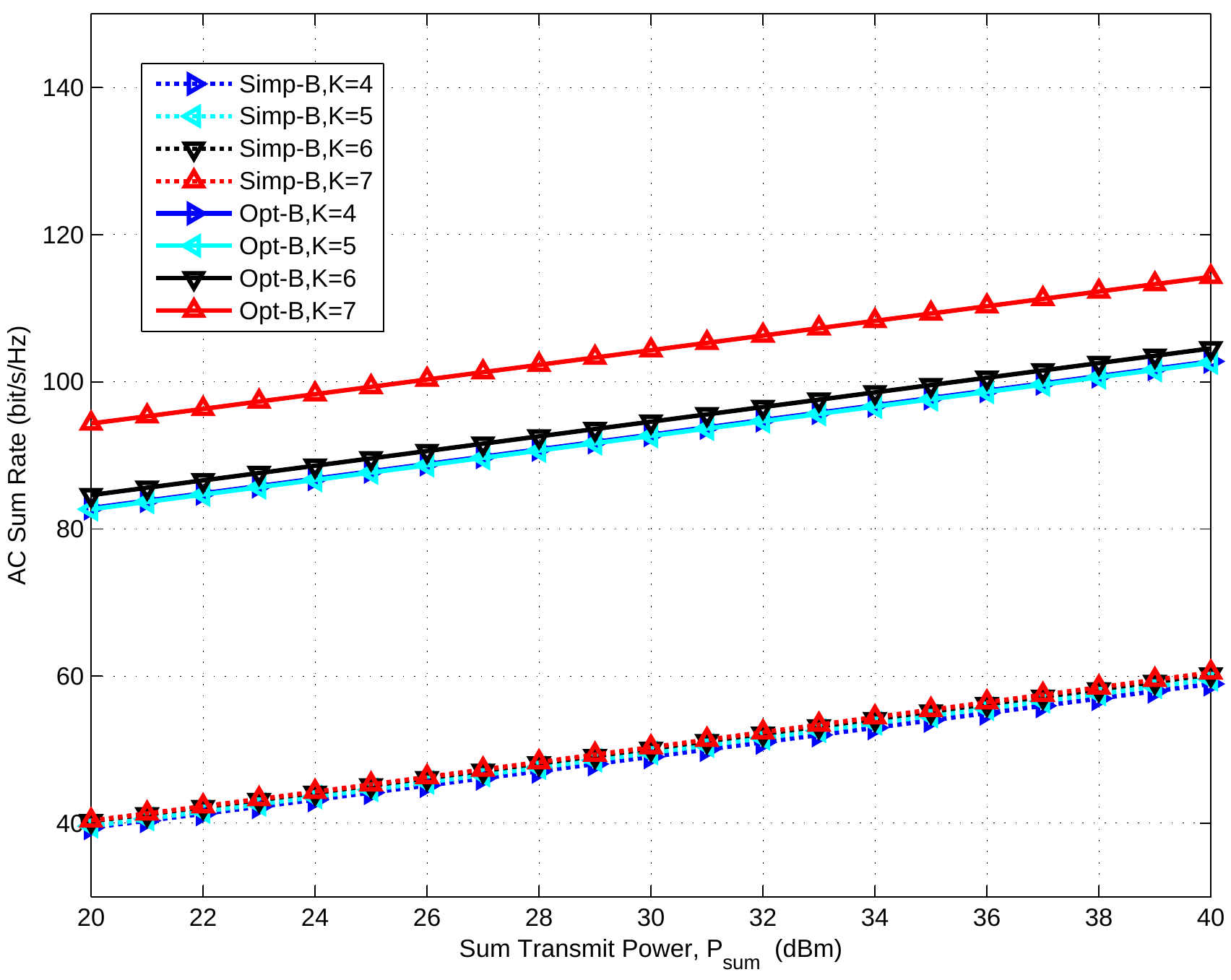}
\caption{AC sum rate versus sum transmit power, where a simple power allocation policy is employed with $\mu=0.5$.}
\label{simulation3}
\end{figure}

\begin{figure}[!tbp]
\centering
\includegraphics[width=0.8\textwidth]{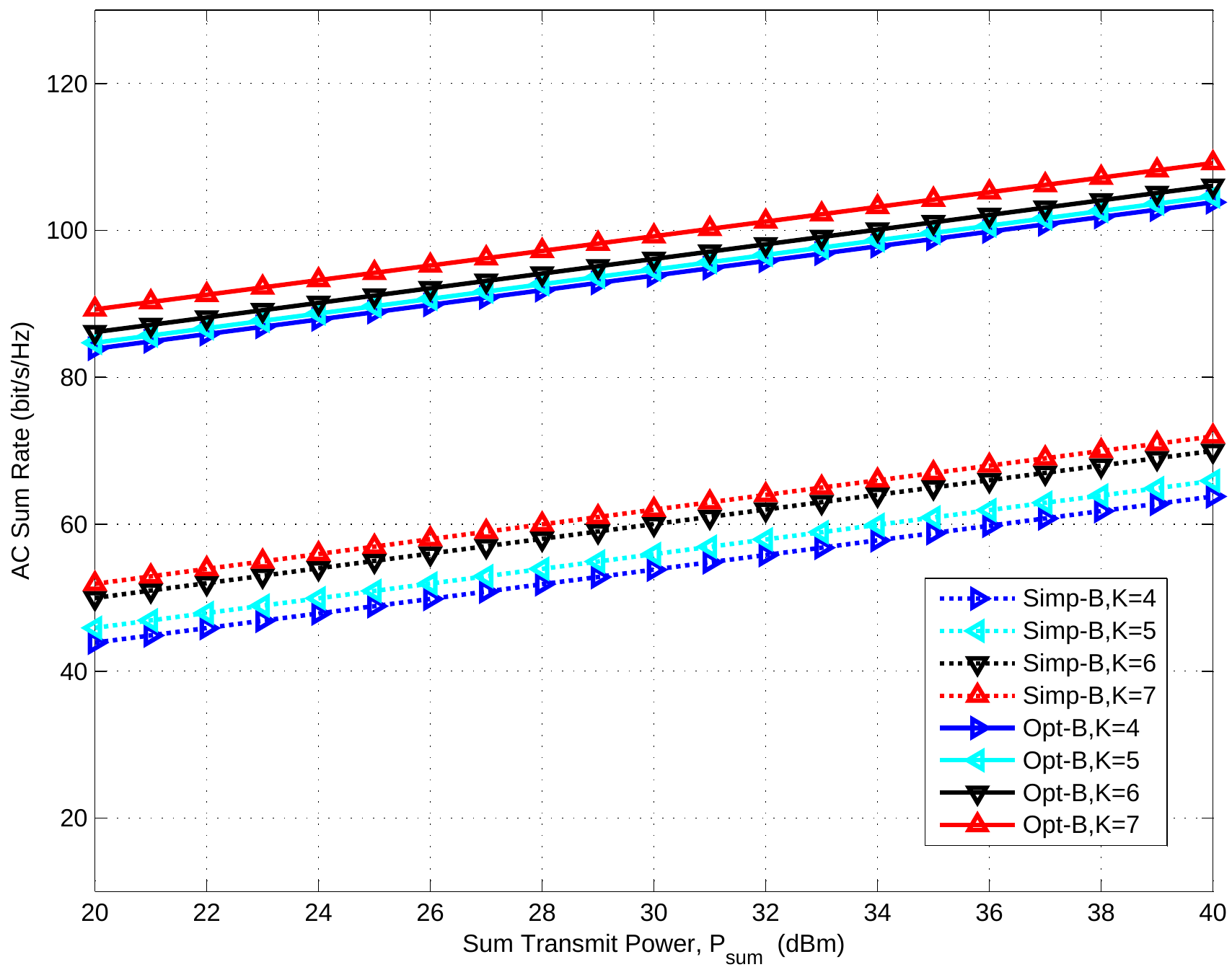}
\caption{AC sum rate versus sum transmit power, where a simple power allocation policy is employed with $\mu=1.5$.}
\label{simulation4}
\end{figure}

\begin{figure}[!tbp]
\centering
\includegraphics[width=0.9\textwidth]{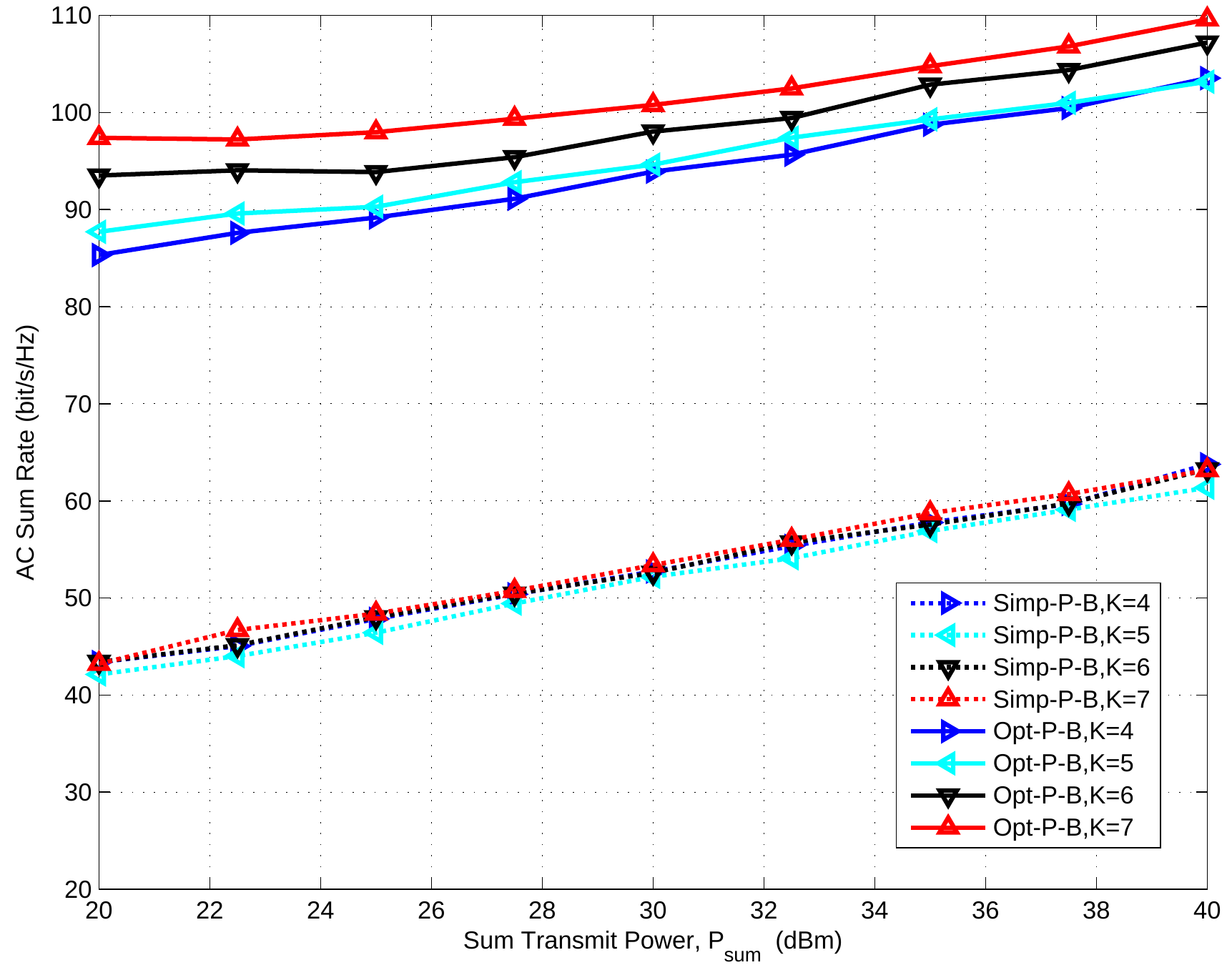}
\caption{AC sum rate versus sum transmit power, where  optimal power allocation  and  beam allocation policies are employed.}
\label{simulation5}
\end{figure}

    In Fig. \ref{simulation1}, we show the AC sum rate of the proposed PDMA approach. Also illustrated in the simulations as performance benchmarks are the AC sum rates of the OMA and PD-NOMA approaches.
    As for the OMA approach, the number of users is set to  3, i.e., each user monopolizes one beam.
    As for the PD-NOMA approach, the number of users is set to  6, i.e., every two users share one beam \cite{Saito2013Non}.
    As for the PDMA approach, the simple power allocation  and beam allocation policies are employed.
    The sum transmit power is set to  \(30\text{dBm}\) in the figure.
    Firstly, it can be observed from the figure that the performance of the proposed PDMA approach is better than those of the OMA and PD-NOMA approaches when the overload ratio $\lambda>1$.
     {Note here that a relatively higher computational complexity as the price will be paid for the better sum rate of our proposed approach compared with the other ``simple" or non-joint design.}
    Secondly, it can be observed that the performance  deteriorates drastically when the power gain factor $\mu$ decreases from 1 to 0, and the reason is that the proposed PDMA approach prefers user fairness to system performance if $\mu<1$ .
    It can also be observed that the performance deteriorates slowly when $\mu$ increases from 1, and the reason is that the proposed PDMA approach prefers system performance to user fairness if  $\mu>1$.
    When $\mu$ increases from 1, less power is allocated to the weaker user. Correspondingly, the performance of the user targeted by the beamforming  deteriorates drastically, which impairs the system performance.
    Thirdly, it can be observed that the performance  is not improved remarkably when $K$ is increased.
    The reason is that the weaker user is still targeted by the beamforming even when $K$ is increased.
    It shows the important role of the weaker user targeted by the beamforming on  system performance.

    In Fig. \ref{simulation2}, we compare the AC sum rates between the proposed PDMA approach employing a simple power allocation policy with $\mu=0.5$ and the proposed PDMA approach employing an optimal power allocation policy, where a simple beam allocation policy is employed.
    We can see that the performance of the proposed PDMA approach is improved slightly when power allocation is optimized.

    In Fig. \ref{simulation3} and Fig. \ref{simulation4}, we compare the AC sum rates between the proposed PDMA approach employing a simple beam allocation policy and the proposed PDMA approach employing an optimal beam allocation policy,
    where a simple power allocation policy is employed with $\mu = 0.5$  in Fig. \ref{simulation3} and $\mu = 1.5$ in Fig. \ref{simulation4}.
    It can be observed that the PDMA approach employing the optimal beam allocation policy  achieves great performance gain over the PDMA approach employing the simple beam allocation policy,
    which verifies the effectiveness of the proposed beam allocation policy.
    We can see from the figures that the performance of the proposed approach with $\mu=1.5$ in Fig. \ref{simulation4} varies in a wider range than that with $\mu=0.5$ in Fig. \ref{simulation3} when the simple beam power allocation policy is utilized,
    which reveals that the impact of stronger users on system performance is larger when beam allocation is not optimized.
    We can also see from the figures that the performance of the proposed approach with $\mu=0.5$ in Fig. \ref{simulation3} varies in a wider range than the scheme with $\mu=1.5$ in Fig. \ref{simulation4} when the optimal beam power allocation policy is utilized,
    which reveals that the impact of weaker users on system performance is larger when beam allocation is optimized. 

    In Fig. \ref{simulation5}, we compare the AC sum rates between the proposed PDMA approach employing  simple power allocation and  beam allocation policies and the proposed PDMA approach employing  optimal power allocation and  beam allocation polices.
    We can see from the figure that the  performance of the proposed PDMA approach is improved greatly when both power allocation and beam allocation are optimized.
    It can be seen from  Fig. \ref{simulation2} to Fig. \ref{simulation5}  that the optimization of beam allocation yields a significant improvement of system performance than the optimization of power allocation.

\section{Conclusions}

    In this paper, we have proposed a  PDMA  joint design approach based on  both the power  domain and beam domain,
    which can be treated  as a more general framework  of the PD-NOMA and OMA approaches.
    The proposed PDMA approach enables the integration of a LAS into multiple access naturally.
    By using the optimized  power allocation and beam allocation policies, great performance gain has been achieved.
     {For the future work, we would like to explore the sensitivity of our proposed approach to CSI estimation and feedback errors.}


\begin{appendices}

\section{Proof of Theorem 1}

    The optimization problem in (\ref{OP}) can be reshaped as a standard form problem \cite{Boyd2004Convex} as follows
\begin{equation}\begin{split}\label{OP2}
\quad & \mathop {\min }\limits_{\boldsymbol{P}} \mathop {}\nolimits^{} \mathop f{\left( \boldsymbol{P} \right)} = - \sum\limits_{n = 1}^N {\sum\limits_{k = 1}^K {{{\log }_2}\left( {1 + \frac{{{\left| h_{nk}\right|}^2{b_{nk}p_{nk}}}}{{1 + {\left| h_{nk}\right|}^2 w_{nk} }}} \right)} } \\
\text{s.t.} \  & \text{C1:}\ g_1^{nk}{\left( \boldsymbol{P} \right)} = {\delta_{nk}} - {p _{nk}} \le 0{, } \ n=1,2,\cdots,N, \ k=1,2,\cdots,K, \\
\quad & \text{C2:}\ g_2{\left( \boldsymbol{P} \right)} = \sum\limits_{n = 1}^N {\sum\limits_{k  = 1}^K {{p_{nk}}} } - {P_{{\text{sum}}}} \le 0{, } \\
\quad & \text{C3:}\ g_3^{nk}{\left( \boldsymbol{P} \right)} = {R_{nk}^{{\text{min}}}} - {\log _2}\left( {1+\frac{{{\left| h_{nk}\right|}^2{b_{nk}p_{nk}}}}{{1 + {\left| h_{nk}\right|}^2w_{nk} }}} \right) \le 0{, } \ n=1,2,\cdots,N, \ k=1,2,\cdots,K,
\end{split}\end{equation}
    where
${w_{nk}} = \left\{ {\begin{array}{*{20}{l}}
{\sum\limits_{k' = k + 1}^K {{b_{nk'}p_{nk'}}}{, }} & {n=1,2,\cdots,N, \ k = 1{, }2, \cdots {, }K - 1},\\
{0{, }}  & {n=1,2,\cdots,N, \ k  = K}.
\end{array}} \right.$

    Next, we analyze the properties of the functions $f{\left( \boldsymbol{P} \right)}$, $g_1^{nk}{\left( \boldsymbol{P} \right)}$, $g_2{\left( \boldsymbol{P} \right)}$ and $g_3^{nk}{\left( \boldsymbol{P} \right)}$.
    It is obvious that  $g_1^{nk}{\left( \boldsymbol{P} \right)}$ is affine in $\boldsymbol{P}$ for any $n$ and $k$ and that  $g_2{\left( \boldsymbol{P} \right)}$ is affine in $\boldsymbol{P}$ as well.
    Due to the independence of $|h_{nk}|^2$ and $w_{nk}$ on $p_{nk}$, $g_3^{nk}{\left( \boldsymbol{P} \right)}$ is convex in $\boldsymbol{P}$ for any $n$ and $k$.
    As for $f{\left( \boldsymbol{P} \right)}$, it can be readily seen that $w_{nk}$ couples multiple variables with respect to $k$ but not $n$.
    Then, according to the transitivity of the convexity [20], the convexity of $f {\left( \boldsymbol{P} \right)}$ can be derived from the convexity of the function $f {\left( \boldsymbol{p}_n \right)} = - \sum\limits_{k = 1}^K {{{\log }_2}\left( {1 + \frac{{{\left| h_{nk}\right|}^2{b_{nk}p_{nk}}}}{{1 + {\left| h_{nk}\right|}^2{w_{nk}}}}} \right)}$ with ${\boldsymbol{p}_n} = \left[ \boldsymbol{P} \right]_n^T$ for $n=1,2,\cdots,N$.

    Now we investigate the convexity of  $f {\left( \boldsymbol{p}_n \right)}$.
    Let ${\nabla^2}f{\left( {\boldsymbol{p}_n} \right)}$ denote the Hessian matrix of  $f{\left( {\boldsymbol{p}_n} \right)}$ whose elements are given by (\ref{2pd}) at the bottom of next page.
    Define
\setcounter{equation}{17}
\begin{align}
{\alpha_0} &\buildrel \Delta \over = \frac{1}{{{{\left( {\frac{1}{{{\left| h_{n1}\right|}^2}} + \sum\limits_{l = 1}^K {{b_{nl}p_{nl}}} } \right)}^2}\ln 2}}{,} \\
{\beta_k}& \buildrel \Delta \over = \frac{1}{{\ln 2}}\sum\limits_{k' = 1}^{k} {\left( {\frac{1}{{{{\left( {\frac{1}{{{\left|h_{n(k' + 1)}\right|}^2}} + {w_{nk'}}} \right)}^2}}} - \frac{1}{{{{\left( {\frac{1}{{{\left| h_{nk'}\right|}^2}} + {w_{nk'}}} \right)}^2}}}} \right)}, \ k =1, 2, \cdots, K-1{.}
\end{align}
Then, the Hessian matrix can be expressed as follows
\begin{equation}\begin{split}\label{hessian}
{\nabla^2}f{\left( {{\boldsymbol{ p}_n}} \right)} =
\left[
\begin{array}{*{20}{c}}
{{\alpha_0}} & {{\alpha_0}} & {{\alpha_0}} & {\cdots} & {{\alpha_0}} \\
{{\alpha_0}} & {{\alpha_0}+{\beta_1}} & {{\alpha_0}+{\beta_1}} & {\cdots} & {{\alpha_0}+{\beta_1}} \\
{{\alpha_0}} & {{\alpha_0}+{\beta_1}} & {{\alpha_0}+{\beta_2}} & {\cdots} & {{\alpha_0}+{\beta_2}} \\
{\vdots} & {\vdots} & {\vdots} & {\ddots} & {\vdots} \\
{{\alpha_0}} & {{\alpha_0}+{\beta_1}} & {{\alpha_0}+{\beta_2}} & {\cdots} & {{\alpha_0}+{\beta_{K-1}}}
\end{array}
\right], \ n =1,2,\cdots,N{.}
\end{split}\end{equation}
    Recall that the property of SIC: \(\left|{h_{nk}}\right| \le \left|{h_{nk'}}\right|\) holds for any \(n\) if \(1 \le k \le k' \le K\).
    Then,
    it can be readily obtained that that $\alpha_0 > 0$ and $\beta_k \ge 0$ for \(1 \le k \le K-1\).
    Correspondingly, we have ${\nabla ^2}f\left( {{\boldsymbol{p}_n}} \right) \succeq {\bf{0}}$.
    Therefore, the function $f{\left( \boldsymbol{p}_n \right)}$ is convex in $\boldsymbol{p}_n$ for any $n=1,2,\cdots,N$.

    According to the above discussions, the optimization problem in (\ref{OP}) can be readily proved to be convex. This completes the proof.

\newcounter{TempEqCnt}
\setcounter{TempEqCnt}{\value{equation}}
\setcounter{equation}{2}
\begin{figure*}[!b]
\hrulefill
\vspace*{4pt}
\setcounter{equation}{16}
\begin{multline}\label{2pd}
{\nabla ^2}f{\left( {{{\boldsymbol p}_n}} \right)_{ij}} = \\
\left\{ {\begin{array}{*{20}{l}}
{\frac{1}{{{{\left( {\frac{1}{{{\left| h_{n1}\right|}^2}} + \sum\limits_{k' = 1}^K {{ b_{nk'}p_{nk'}}} } \right)}^2}\ln 2}}\text{, } }&{   {i = 1\text{, }\forall j,}  \ \text{or} }\\
{} &{ {j = 1\text{, }\forall i}  ,}\\
{\frac{1}{{{{\left( {\frac{1}{{{\left| h_{n1}\right|}^2}} + \sum\limits_{k' = 1}^K {{ b_{nk'}p_{nk'}}} } \right)}^2}\ln 2}} + \frac{1}{{\ln 2}}\sum\limits_{k' = 1}^{\min \left( {i\text{, }j} \right) - 1} {\left( {\frac{1}{{{{\left( {\frac{1}{{{\left| h_{n(k' + 1)} \right|}^2}} + {w_{nk'}}} \right)}^2}}} - \frac{1}{{{{\left( {\frac{1}{{{\left| h_{nk'}\right|}^2}} + {w_{nk'}}} \right)}^2}}}} \right)}\text{, }  }&{ i \ne 1\text{, } j \ne 1.}
\end{array}} \right.
\end{multline}
\setcounter{equation}{\value{TempEqCnt}}
\end{figure*}

\section{Proof of Theorem 2}
     {To prove that $\bar{\boldsymbol{B}}$ can be set to  $\boldsymbol{I}$ if the BF and SF matrices are constructed based on the ZF
    or MMSE criteria,}
    we take the case for ZF-BF and ZF-SF as examples, and the case for  MMSE-BF and  MMSE-SF is similar.
     {Firstly, by using elaborate mathematical manipulations, the considered channel matrix, the ZF-BF matrix, and the ZF-SF matrix can be reexpressed neatly. Then, by analyzing the relevant items of the equivalent normalized channel gain in \eqref{h_nk} with the reexpressed matrices, particular characteristics can be obtained. Correspondingly, the conclusion can be readily achieved. We present the proof in detail as follows.}

    Let ${\boldsymbol{G}_C} = {\left[ {\boldsymbol{G}_{{1_{{\rm{tar}}}}}^H,\boldsymbol{G}_{{2_{{\rm{tar}}}}}^H, \cdots ,\boldsymbol{G}_{{n_{{\rm{tar}}}}}^H, \cdots ,\boldsymbol{G}_{{N_{{\rm{tar}}}}}^H} \right]^H}$ denote the composite channel matrix of the target users. 
    Then, ${\boldsymbol{G}_{{n_{{\rm{tar}}}}}} = {\boldsymbol{X}_{n_{{\rm{tar}}}}} {\boldsymbol{G}_C}$ for user ${n_{{\rm{tar}}}}$ with $ \boldsymbol{X}_{n_{{\rm{tar}}}} = \boldsymbol{e}_{n_{{\rm{tar}}}}^H \otimes {{\boldsymbol{I}}_{{N_R}}}$.
    As for ZF-BF, let ${\boldsymbol{F}_C} \in {\mathbb{C}^{{N_T} \times {N_R}N}}$ denote the composite ZF-BF matrix. Then, it can be expressed as ${\boldsymbol{F}_C} = \boldsymbol{G}_C^H{\left( {{\boldsymbol{G}_C}\boldsymbol{G}_C^H} \right)^{ - 1}}$ if ${N_R}N < {N_T}$ or ${\boldsymbol{F}_C} ={\left( {{\boldsymbol{G}_C^H}\boldsymbol{G}_C} \right)^{ - 1}}\boldsymbol{G}_C^H$ otherwise.
    Because ${N_R}N < {N_T}$ is always guaranteed by the LSA, we take ${\boldsymbol{F}_C} = \boldsymbol{G}_C^H{\left( {{\boldsymbol{G}_C}\boldsymbol{G}_C^H} \right)^{ - 1}}$ for example.
    It can be reshaped in a partitioned structure ${\boldsymbol{F}_C} = \left[ {\boldsymbol{F}_1^{},\boldsymbol{F}_2^{}, \cdots ,{\boldsymbol{F}_n}, \cdots ,\boldsymbol{F}_N^{}} \right]$.
    Correspondingly, the BF vector ${\boldsymbol{f}_n} $ can be expressed in the following form: ${\boldsymbol{f}_n} = {\boldsymbol{F}_n}{{\bf{1}}^{{N_R} \times 1}}$.
    Therefore, the ZF-BF matrix can be expressed as $\boldsymbol{F} = \left[ {\boldsymbol{f}_1^{},\boldsymbol{f}_2^{}, \cdots ,{\boldsymbol{f}_n}, \cdots ,\boldsymbol{f}_N^{}} \right] = {\boldsymbol{F}_C}{\boldsymbol{Y}}$ with $\boldsymbol{Y}  = {{\boldsymbol{I}}_{N}} \otimes {{\boldsymbol{1}}_{N_R}}$.
    As for ZF-SF, the ZF-SF matrix of the $k$-th user can be expressed as ${\boldsymbol{V}_k} = {\boldsymbol{G}_k}\boldsymbol{F}{\left( {{\boldsymbol{F}^H}\boldsymbol{G}_k^H{\boldsymbol{G}_k}\boldsymbol{F}} \right)^{ - 1}}$ for $k=1,2,\cdots,K$.

    Through proper mathematical manipulations, we have
    \begin{equation}
    \boldsymbol{v}_{nk}^H{\boldsymbol{G}_k}{\boldsymbol{f}_{n'}} = \boldsymbol{e}_n^H{\left( {{\boldsymbol{F}^H}\boldsymbol{G}_k^H{\boldsymbol{G}_k}\boldsymbol{F}} \right)^{ - 1}}{\boldsymbol{F}^H}\boldsymbol{G}_k^H{\boldsymbol{G}_k}\boldsymbol{Fe}_{n'}^{} = \boldsymbol{e}_n^H\boldsymbol{e}_{n'}^{}, \
    n, n'= 1,2,\cdots, N, \ k=1,2,\cdots,K.
    \end{equation}
    Then, the expression in (\ref{h_nk}) can be simplified to be ${\tilde{h}_{nk}} =  {\frac{1}{\sqrt{\sigma _k^2{\xi_{nk}}}}}$, where
    \begin{equation}
    {\xi_{nk}} \buildrel \Delta \over = \boldsymbol{v}_{nk}^H{\boldsymbol{v}_{nk}} = \boldsymbol{e}_n^H{\left( {{\boldsymbol{F}^H}\boldsymbol{G}_k^H{\boldsymbol{G}_k}\boldsymbol{F}} \right)^{ - 1}}\boldsymbol{e}_n^{}.
    \end{equation}
    Define $\boldsymbol{T}_k = {{\boldsymbol{F}^H}\boldsymbol{G}_k^H{\boldsymbol{G}_k}\boldsymbol{F}}$. Then,  for user $n_{\rm{tar}}$, $\boldsymbol{T}_{n_{\rm{tar}}}$ can be reshaped as follows 
\begin{equation}\begin{split}\label{T_k}
{\boldsymbol{T}_{{n_{{\rm{tar}}}}}}  
&
={\boldsymbol{Y} ^H}\boldsymbol{F}_C^H\boldsymbol{G}_C^H{\boldsymbol{X}_{n_{{\rm{tar}}}} ^H}\boldsymbol{X}_{n_{{\rm{tar}}}} {\boldsymbol{G}_C}{\boldsymbol{F}_C}\boldsymbol{Y} \\
&={\boldsymbol{Y} ^H}{\boldsymbol{X}_{n_{{\rm{tar}}}} ^H}\boldsymbol{X}_{n_{{\rm{tar}}}} \boldsymbol{Y} \\
&\triangleq {\boldsymbol{\Pi}_{n_{{\rm{tar}}}} }.
\end{split}\end{equation}
According to the above definitions of $\boldsymbol{X}_{n_{{\rm{tar}}}}$ and  $\boldsymbol{Y}$,   it can be readily obtained that  all the elements of the square matrix $\boldsymbol{\Pi}_{n_{{\rm{tar}}}}$ are zeros except  ${\left[ \boldsymbol{\Pi}_{n_{{\rm{tar}}}}  \right]_{{n_{{\rm{tar}}}}{n_{{\rm{tar}}}}}} = {1}$.
    Therefore, when ${\tilde{h}_{n n_{\rm{tar}}}}$ is calculated through the  expression in (\ref{h_nk}) for user $n_{\rm{tar}}$ on all the beams except the $n_{\rm{tar}}$-th beam, the noise term will be unusually magnified and  ${\tilde{h}_{n n_{\rm{tar}}}}$ will be approximately equal to zero for any $n=1,2,\cdots,N$ and $n \neq n_{\rm{tar}}$, which means that
the $n_{\rm{tar}}$-th user may be located in the coverage edge of the $n$-th beam.
    Obviously, neither power resources nor  beam resources allocated to the $n_{\rm{tar}}$-th user on the $n$-th beam can contribute much to  system performance.
    Correspondingly, the relevant  $\tilde{b}_{n n_{\rm{tar}}} $ can be set to zero if $\tilde{h}_{n n_{\rm{tar}}} \approx 0$.
    Then, according to the above discussions, $\bar{\boldsymbol{B}} = \boldsymbol{I}$ can be readily obtained. This completes the proof.

\end{appendices}

%

\bibliographystyle{IEEEtran}
\bibliography{manuscript_reference}

\end{document}